\documentclass[a4paper,reqno]{amsart}
\usepackage[utf8x]{inputenc}
\usepackage{graphicx,color}
\usepackage{amssymb,amsmath}
\newtheorem{theorem}{Theorem}[section]

\newtheorem{lemma}[theorem]{Lemma}

\usepackage{enumerate}

\DeclareMathOperator{\Ov}{Ov}
\DeclareMathOperator{\MaxOv}{MaxOv}
\DeclareMathOperator{\MinOv}{MinOv}

\newcommand{\chull}[1]{\langle #1 \rangle}

\begin{document}
\title[Creating and controlling overlap]{Creating and controlling overlap in two-layer networks. Application to a mean-field SIS epidemic model with awareness dissemination.}
\author[D. Juher]{David Juher}
\address{Dept. IMAE,
Universitat de Girona, Catalonia}
\email{david.juher@udg.edu}
\author[J. Salda\~na]{Joan Salda\~na}
\address{Dept. IMAE,
Universitat de Girona, Catalonia}
\email{joan.saldana@udg.edu}

\begin{abstract}
We study the properties of the potential overlap between two networks $A,B$ sharing the same set of $N$ nodes (a two-layer network) whose respective degree distributions $p_A(k), p_B(k)$ are given. Defining the overlap coefficient $\alpha$ as the Jaccard index, we prove that $\alpha$ is very close to 0 when $A$ and $B$ have been independently generated via the configuration model algorithm. We also derive an upper bound $\alpha_M$ for the maximum overlap coefficient permitted in terms of $p_A(k)$, $p_B(k)$ and $N$. Then we present an algorithm based on cross-rewiring of links to obtain a two-layer network with any prescribed $\alpha$ inside the range $(0,\alpha_M)$. Finally, to illustrate the importance of the overlap for the dynamics of interacting contagious processes, we derive a mean-field model for the spread of an SIS epidemic with awareness against infection over a two-layer network, containing $\alpha$ as a parameter. A simple analytical relationship between $\alpha$ and the basic reproduction number follows. Stochastic simulations are presented to assess the accuracy of the upper bound $\alpha_M$ and the predictions of the mean-field epidemic model.
\end{abstract}

\maketitle

\section{Introduction}
\label{s1}

Some contagious processes interact with each other during their propagation, which can occur either through the same route of transmission or through routes that share the same set of nodes but use different types of connections. In the second case, the description of the spread uses the concept of multilayer or multiplex network, namely, a set of nodes (individuals, computers, etc.) connected by qualitatively different types of links corresponding to possible relationships among them (acquaintanceship, friendship, physical contact, social networks, etc), each layer defined by a type of connection. Competitive viruses spreading simultaneously through different routes of transmission over the same host population, or the spread of a pathogen and awareness during an epidemic episode are examples of processes that are better described by means of multilayer networks \cite{Sahneh-Scoglio}.

In the last years it has been a development of the mathematical formulation of multiplex networks and, also, of more general interconnected networks for which the set of nodes does not need to be the same at each layer \cite{DD,DS,Saumell}. Moreover, recent results show the importance of the interrelation between different layers in determining the fate of competitive epidemic processes \cite{Funk10,Sahneh-Scoglio}. In other cases, however, the  importance of such an interrelation is not so evident from the analytical results of the epidemic threshold  \cite{Granell13,Granell14}, or even seems to be not relevant at all \cite{Wei}.

Only a few papers dealing with competing epidemics over multilayer networks focus on the impact of layer overlap on the epidemic dynamics \cite{Funk09,Funk10,Marceau}. In \cite{Funk10}, the authors consider a sequential propagation of two epidemics using distinct routes of transmission over a network consisting of two partly overlapped layers. Using bond percolation, it is determined the success of a second epidemic through that part of its route of transmission whose nodes have not been infected by the first epidemic. In \cite{Marceau}, the authors develop an analytical approach to deal with simultaneous spread of two interacting viral agents on two-layered networks. In this work, moreover, the respective effects of overlap and correlation of the degrees of nodes in each layer on the epidemic dynamics are considered.

Here the \emph{overlap} $\alpha$ between two (labeled) networks $A$ and $B$ of $N$ nodes is defined as the fraction of links of the union network that are common links of $A$ and $B$ or, equivalently, the probability that a randomly chosen link of the network $A\cup B$ is simultaneously a link of both $A$ and $B$. In fact $\alpha$ is the \emph{Jaccard index} as defined in \cite{jac}. Just to illustrate that this simple statistical parameter can play a critical role in the qualitative response of a two-layer network model, in Sect.~\ref{s2} we present a mean-field model for the spread of two contagious process interacting each other, namely, the spread of an infectious agent and the raising in awareness of preventive behaviours. As an interesting feature, the overlap coefficient between the networks embedding the respective routes of transmission is a parameter of the model. This allows us to derive a simple relationship between $\alpha$ and the epidemic threshold. Provided that one wants to perform simulations to validate this (or any) model, a systematic procedure to generate couples of networks of given size and degree distributions with a prescribed value of $\alpha$ would be a useful tool.

However, the following natural question arises: \textit{Given respective degree distributions $p_A(k)$ and $p_B(k)$ for each network layer, which is the range of attainable overlap coefficients between them?} In previous papers dealing with this issue \cite{Funk10,Marceau}, a joint degree distribution $\rho(k_1, k_2, k_3)$ is considered to generate a two-layer network with arbitrary overlap by decomposing it into three non-overlapped networks. The third marginal degree distribution is the one for the overlapped part of the two layers, whereas the other two correspond to the non-overlapped parts of each layer. Therefore, the probability that a randomly selected node has degree $k_1$ on the first layer and degree $k_2$ on the second one is given by the joint degree distribution $P(k_1,k_2)$ obtained from $\rho$ as $P(k_1, k_2)  = \sum_{k_3} \rho(k_1 - k_3, k_2 - k_3, k_3)$. In other words, the overlap between both layers is prescribed before hand by $\rho(k_1,k_2,k_3)$. In contrast, our approach is based on the study of the potential overlap between two networks whose (finite, empirical) degree distributions are previously fixed. More precisely, in Sect.~\ref{s3} and \ref{s4} we estimate the minimum and maximum values (call them $\alpha_m$ and $\alpha_M$) for the overlap coefficient between two networks of size $N$ and degree distributions $p_A(k)$ and $p_B(k)$. In Sect.~\ref{s5} we present an algorithm that takes as input $N$, $p_A(k)$, $p_B(k)$ and $\alpha\in(\alpha_m,\alpha_M)$, and generates a couple of networks of $N$ nodes, with respective degree distributions $p_A(k)$ and $p_B(k)$ and overlap coefficient close to $\alpha$. So we are given a tool to test the analytical predictions relating overlap and epidemic thresholds. Finally, in Sect.~\ref{s6} we assess the accuracy of the predictions of the mean-field formulation by comparing them to stochastic simulations of the contagious processes over complex random networks.

\section{Motivation of the problem: a mean-field SIS epidemic model defined on a two-layer network}\label{s2}
We start this section by fixing some terminology. All along this paper, the nodes of any network will be labeled with the natural numbers $\{1,2,\ldots,N\}$.
The cardinality of a finite set $X$ will be denoted by $|X|$. Let $V=\{1,2,\ldots,N\}$ for some $N\in\mathbb{N}$. Let $E$ and $E'$ be two subsets
of $\{\{i,j\}:i\ne j\mbox{ and }i,j\in V\}$. Let $G$ and $G'$ be the undirected networks having $V$ as the set of nodes and $E$ and $E'$ as the respective sets of links. The \emph{union network} $G\cup G'$ is the undirected network whose sets of nodes and links are $V$ and $E\cup E'$ respectively. By definition, we will say that $G$
and $G'$ are \emph{different} from each other if and only if $E\ne E'$. In particular, if we have a network $H$ and we simply permute the labels of the nodes of
$H$, then we obtain a network that is in general different from (but isomorphic to) $H$. Observe that the union operation is not a topological invariant: the union of two networks does not depend only on their shapes but also on the way their nodes are labeled.
The \emph{overlap} between $G$ and $G'$ is defined as the fraction
\[ \Ov(G,G'):=\frac{|E\cap E'|}{|E\cup E'|}=\frac{|E\cap E'|}{|E|+|E'|-|E\cap E'|},  \]
which can be thought as the probability that a randomly chosen link of $G\cup G'$ is simultaneously a link of both $G$ and $G'$.

A \emph{degree set} of cardinality $N$ is a multiset (i.e. multiple instances of each element are allowed) of $N$ integers that is realizable as the set of degrees of a network.
That is, there exist a labeling $\{k_1,k_2,\ldots,k_N\}$ of the elements of the set and a network $G$ of $N$ nodes such that $k_i$ is the degree of the node $i$. The ordered list
$(k_1,k_2,\ldots,k_N)$ will be called the \emph{degree sequence} of $G$. A probability distribution $p(k)$ with bounded support will be called \emph{empirical (of $N$ nodes)} if it is
realizable as the degree distribution of a network of $N$ nodes. That is, there exists a network $G$ of $N$ nodes such that:
\begin{enumerate}
\item[(S1)] The degree set $\{k_1,k_2,\ldots,k_N\}$ of $G$ satisfies the well-known Havel-Hakimi condition \cite{havel,hakimi}
\item[(S2)] $N_k:=|\{i:k_i=k\}|=p(k)N$
\item[(S3)] $\sum k_i=:2L$ is even
\item[(S4)] If $\chull{k}$ denotes the expected degree of a node, then $\chull{k}N=2L$.
\end{enumerate}
We use the term \emph{empirical} for a degree distribution to distinguish it from a (theoretical, not necessarily with bounded support) probability distribution $p(k)$. In this case, for
any $N\in\mathbb{N}$, one can use several standard algorithms (see Sect.~\ref{s3}) to construct a network $G_N$ of $N$ nodes whose empirical degree distribution $p_N(k)$ is close to $p(k)$, in the sense that,
for big enough values of $N$, $p_N(k)$ converges in probability to $p(k)$ (\cite{Britton}, Theorem 2.1).

\subsection{The model}
Epidemic models describe the spread of infectious diseases on populations whose individuals are classified into distinct classes according to their infection state as, for instance, the class of susceptible (S) individuals and the class of infectious (I) ones. A closer look at the physical transmission of an infection reveals that a suitable description of populations must take into account the network $A$ of physical contacts among individuals, with nodes representing individuals and links corresponding to physical contacts along which disease can propagate. On the other hand, if one assumes that the probability of getting infected through an infectious contact S-I depends on the awareness state of the susceptible individual, then a second network $B$ over which information about the infection state of individuals circulates can be considered. This dissemination network shares the same set of nodes with the one of physical contacts, but has a different set of links representing, for instance, relationships with friends and acquaintances. So, if a pair of individuals, one susceptible and the other infectious, are connected to each other on both networks, one can assume that the transmission rate $\beta_c$ (here $c$ stands for \emph{common}) will be smaller than the normal transmission rate $\beta$. This is because susceptible individuals have information about the health state of their infected partners and react by adopting preventive measures to diminish the risk of contagion.

According to this scenario, next we derive a mean-field susceptible-infectious-susceptible (SIS) epidemic model which implicitly assumes spreading of both information and an infectious agent over a two-layer network. Following the standard approach for sexually transmitted diseases (STDs) where the heterogeneity in the number of contacts (sexual partners) is a basic ingredient \cite{Anderson}, individuals are classified according to their infection state and  their number of physical contacts. So, the model will take into account the network layer $A$ of physical contacts in terms of its degree distribution $p_A(k)=N_k/N$ where $N_k$ is the number of individuals having degree $k$. Analogously, the dissemination network (network layer $B$) is described by its degree distribution $p_B(k)$. A key assumption in the model derivation is the existence of a \textit{partial} and \textit{uniform} overlap between the links of each layer, which means that the probability of finding two nodes connected to each other in both networks does not depend on the degrees of the pair. For sake of brevity, a pair of such nodes is said to share a \textit{common link}, although the natures of the connections are dissimilar.

Within each layer, it is assumed that there is no degree-degree correlation, i.e., neighbours in each layer are randomly sampled from the population according to the so-called proportionate mixing of individuals \cite{Diekmann}. This means that, in each layer, the probability $P(k'|k)$ that a node of degree $k$ is connected to a node of degree $k'$ is independent of the degree $k$ and it is given by the fraction of links pointing to nodes of degree $k'$, i.e., $P(k'|k) = k'p(k') /\chull{k}$ \cite{Diekmann}. Therefore, the expected degree of a node reached by following a randomly chosen link, i.e., the expected degree of a neighbour in a population with proportionate mixing is $\chull{k^2}/\chull{k}$. On the other hand, let $I_k$ be the number of infectious nodes of degree $k$ in network layer $A$. Although the links are unordered pairs of connected nodes by definition, let us consider that every link $\{u,v\}$ gives rise to two \emph{oriented links} $u\rightarrow v$ and $v\rightarrow u$. Then, the probability that a randomly chosen oriented link of $A$ leads to an infectious node is given by the fraction of oriented links in $A$ pointing to infectious nodes, that is,
$$
\Theta_I = \frac{1}{\chull{k_A}N}  \sum_k k\,I_k = \frac{1}{\chull{k_A}}  \sum_k k\,i_k
$$
where $\chull{k_A}$ is the average degree in $A$, and $i_k := I_k/N$ is the fraction of nodes that are both infectious and of degree $k$ in $A$.

Finally, let $L_A$, $L_B$, and $L_{A \cap B}$ denote the number of links of $A$, $B$, and common links, respectively. Let $p_{B|A}$ be the probability that a randomly chosen link of $A$, an $A$-link, connects two nodes that are also connected in $B$, that is, $p_{B|A}=\frac{L_{A \cap B}}{L_A}$. Similarly, $p_{A|B}=\frac{L_{A \cap B}}{L_B}$ is the probability that a randomly chosen $B$-link is a common link to both networks. With all these quantities, the epidemic spreading is described in terms of the following differential equation for the number of infectious nodes of degree $k$ in layer $A$:
\begin{equation} \label{SIS-HMFM}
\displaystyle
\frac{dI_k}{dt} = k (1-p_{B|A}) \beta \, S_k \, \Theta_I + k \, p_{B|A} \, \beta_c S_k \, \Theta_I - \mu I_k
\end{equation}
with $S_k=N_k-I_k$ being the number of susceptible nodes of degree $k$ in layer $A$. Here $\beta$ is the transmission rate through a non-common infectious link, and $\beta_c$ is the transmission rate through a common infectious link.

The first term in the rhs of \eqref{SIS-HMFM} is the rate of creation of new infectious nodes of degree $k$ in $A$ due to transmissions of the infection through links that only belong to layer A, whereas the second term is the rate of creation of new infectious nodes from transmissions across common links. The last term accounts for the recoveries of infectious nodes, which occur at a recovery rate $\mu$. Here $\chull{k_A}p_{B|A}$ is the expected number of common oriented links. Therefore, since this number is the same regardless the network we use to compute it, the following consistency relationship must follow:
\begin{equation}
\chull{k_A}p_{B|A} = \chull{k_B}p_{A|B}.
\label{consistency}
\end{equation}

Now let us express $p_{B|A}$ and $p_{A|B}$ in terms of the overlap $\alpha:=\Ov(A,B)$, which is defined as $\alpha=\frac{L_{A \cap B}}{L_{A \cup B}}$ where $L_{A \cup B}$ is the set of links of the union network $A \cup B$. Using (S4), $p_{B|A}$ can be expressed in terms of $\alpha$ as follows:
$$
p_{B|A} = \frac{L_{A \cap B}}{L_A} = \frac{L_{A \cap B}}{L_{A \cup B}} \, \frac{L_{A \cup B}}{L_A} = \alpha \, \frac{L_A + L_B - L_{A \cap B}}{L_A} = \alpha \, \left( 1+ \frac{\chull{k_B}}{\chull{k_A}} - p_{B|A} \right).
$$
From this simple relationship it immediately follows that
\begin{equation} \label{pB|A}
p_{B|A} = \left(1+\frac{\chull{k_B}}{\chull{k_A}}\right) \frac{\alpha}{1+\alpha}.
\end{equation}
Similarly, we also have that $\displaystyle p_{A|B} = \left( 1 + \frac{\chull{k_A}}{\chull{k_B}} \right) \frac{\alpha}{1+\alpha}$. Note that, as expected, $p_{B|A}$ and $p_{A|B}$ fulfil relationship \eqref{consistency}.

Introducing \eqref{pB|A} into Eq.~\eqref{SIS-HMFM}, the overlap appears as a new parameter of the model which now, in terms of the fraction $i_k$ of nodes that are both infectious and of degree $k$, reads
\begin{equation}
\label{SIS-HMFM2}
\frac{di_k}{dt} = \frac{k}{1+\alpha}  \left( \beta \left(1 - \frac{\chull{k_B}}{\chull{k_A}} \alpha \right) + \beta_c \left( 1 +  \frac{\chull{k_B}}{\chull{k_A}} \right)\alpha \right) (p_A(k) - i_k) \, \Theta_I - \mu i_k.
\end{equation}
This equation corresponds to the standard SIS model for heterogeneous populations with proportionate mixing, but with an averaged transmission rate which depends on $\alpha$.

Simple facts about this equation are:

\begin{enumerate}
\item By Lemma~\ref{uppmax}, an upper bound for the maximum overlap coefficient is given by $\min\{\chull{k_A},\chull{k_B}\}/\max\{\chull{k_A},\chull{k_B}\}$. Since the factor $\alpha/(1+\alpha)$ in \eqref{pB|A} is increasing in $\alpha$, when $\chull{k_A}\le\chull{k_B}$ we get $p_{B|A}\le1$ while, for $\chull{k_A} > \chull{k_B}$, we get $p_{B|A}\le\chull{k_B}/\chull{k_A} < 1$.
\item If $\beta_c=\beta$ or $\alpha=0$, Eq.~\eqref{SIS-HMFM2} reduces to the classic SIS-model, as expected, because information dissemination plays no role in the infection spread. If $\alpha=1$, we actually have one network and again Eq.~\eqref{SIS-HMFM2} reduces to the SIS-model but now with $\beta$ replaced by $\beta_c$.
\end{enumerate}

To determine the impact of the network overlap on the initial epidemic growth, we  linearise \eqref{SIS-HMFM2} about the disease-free equilibrium $i^*_k = 0 \, \forall k$ and obtain that the elements of the Jacobian matrix $J^*$ evaluated at this equilibrium are
$$
J^*_{kk'} = \frac{\beta_0(\alpha)}{\chull{k_A}} k k' p_A(k) - \mu \delta_{kk'}
$$
where $\beta_0(\alpha):=\left( \beta \left(1 - \frac{\chull{k_B}}{\chull{k_A}} \alpha \right) + \beta_c \left( 1 +  \frac{\chull{k_B}}{\chull{k_A}} \right)\alpha \right) /\, (1+\alpha)$ and $\delta_{kk'}$ is the Kronecker delta. Since the dominant eigenvalue of the matrix $(kk'p_A(k))$ is equal to $\chull{k^2_A}=\sum_k k^2 p_A(k)$ (with an associated eigenvector whose components $v_k$ are proportional to $kp_A(k)$), it follows that the dominant eigenvalue of $J^*$ is
$$
\lambda_1=\frac{\chull{k^2_A}}{\chull{k_A}} \beta_0(\alpha) - \mu,
$$
which corresponds to the initial growth rate of the epidemic (cf. \cite{Anderson,May}  for $\alpha=0$). From this expression we get that
$\lambda_1$ decreases with $\alpha$ when $\beta_c < \beta$.

We can also measure the initial epidemic growth in terms of the basic reproduction number $R_0$, i.e., the average number of secondary infections caused by a typical infectious individual at the beginning of an epidemic in a wholly susceptible population \cite{Diekmann}. Interpreting $\beta_0(\alpha)$ as an averaged transmission rate weighted by the overlap coefficient $\alpha$ and recalling that $\chull{k^2_A}/\chull{k_A}$ is the expected degree of a neighbour in a population with proportionate mixing, $R_0$ is given by
$$
R_0 = \frac{\chull{k^2_A}}{\chull{k_A}} \frac{\beta_0(\alpha)}{\mu}
= \frac{\chull{k^2_A}}{\chull{k_A} (1+\alpha) \, \mu} \left(\beta \left(1 - \frac{\chull{k_B}}{\chull{k_A}} \, \alpha \right) + \beta_c \left(1 + \frac{\chull{k_B}}{\chull{k_A}} \right)\alpha \right).
$$
Therefore, as expected from the expression of $\lambda_1$, $R_0$ is a decreasing function of the overlap coefficient between the two layers  as long as $\beta_c < \beta$. Note that this expression of $R_0$ is a straightforward extension of the one obtained in \cite{Anderson} for heterogeneous populations and STDs. Figure~\ref{R0-SIS} shows this relationship when layer $A$ has, for instance, an exponential degree distribution with minimum degree $k_{\min}=10$. For this distribution, $\langle k_A \rangle = 2 k_{\min}$ and $\langle k^2_A \rangle / \langle k_A \rangle =5/2 \cdot k_{\min}$ which amount to the values used in the figure.

\begin{figure}[hf]
\centering
\includegraphics[scale=0.45]{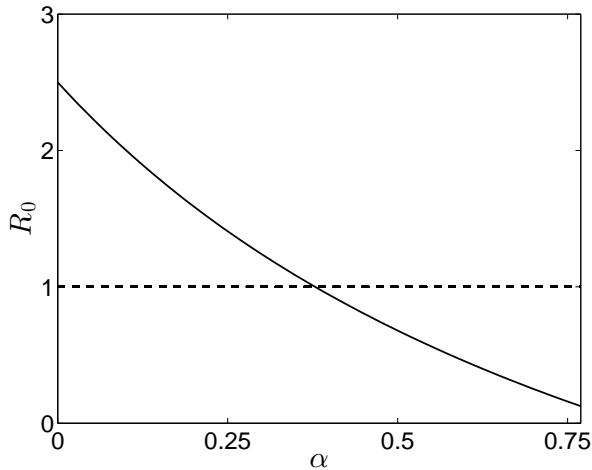}
\caption{Graph of $R_0$ of the SIS model as a function of the overlap coefficient $\alpha$. Parameters: $\mu=1$, $\beta=0.1$, $\beta_c=0.005$, $\langle k_A\rangle = 20$, $\langle k^2_A \rangle = 500$, and $\langle k_B \rangle = 26$. For these mean degrees, $\alpha \in [0, 10/13]$ by Lemma~\ref{uppmax}.
}\label{R0-SIS}
\end{figure}

As usual, it would be desirable to test the accuracy of the model by collating the numerical integration of equations \eqref{SIS-HMFM2} with the output of stochastic simulations. Note that in the derivation of \eqref{SIS-HMFM2} we have assumed the statistical uniformity of several network features. In particular, observe that the entire degree distribution $p_B(k)$ of the dissemination network plays no role in the equations (this is not the case for $p_A(k)$). In fact, the role of layer $B$ is reduced to its mean degree $\chull{k_B}$ via the term $p_{B|A}$. In consequence, it makes sense to perform stochastic simulations with a number of different topologies for $A$ and $B$, in order to evaluate in which situations the mean-field nature of the model fails in giving accurate predictions for the epidemic progression. On the other hand, we are mainly interested in the overlap $\alpha$ as the critical parameter of the model. So, once the empirical degree distributions $p_A(k)$ and $p_B(k)$ are decided, we aim at performing simulations for several values of $\alpha$. Taking it all into account, the following natural questions arise. First, which is the possible range of permitted overlaps between any couple of networks $A,B$ with previously fixed size $N$ and empirical degree distributions $p_A(k)$ and $p_B(k)$? Second, given a value of $\alpha$ inside this range, it is possible to design an algorithm to construct two networks $A$ and $B$ whose degrees are respectively distributed according to $p_A(k)$ and $p_B(k)$ with the prescribed overlap $\alpha$? Both issues are discussed in the following sections.

\section{The expected overlap between two random independent layers}\label{s3}
Assume that we are given two empirical degree distributions $p(k),p'(k)$ of $N$ nodes, with corresponding degree sets $K$ and $K'$.
Let $n$ and $n'$ be the total number of pairwise different networks having respectively $K$ and $K'$ as degree set, each one numbered with an integer in the
range $[1,n]$ (respectively, $[1,n']$). Then we can clearly consider a function of two variables $\Ov(x,y)$ on the grid of all pairs
$(x,y)$ of integers in $[1,n]\times[1,n']$, that gives the value of the overlap of the networks numbered as $x$ and $y$. Observe that the function $\Ov(x,y)$
has a global minimum/maximum. These extremal values will be denoted by $\MinOv(K,K')$ and $\MaxOv(K,K')$, or by $\MinOv_N(p,p')$ and $\MaxOv_N(p,p')$. The
problem of finding or estimating $\MinOv_N(p,p')$ and $\MaxOv_N(p,p')$ naturally arises. Note that a brute force algorithm to compute them
by exploring $\Ov(x,y)$ for all $(x,y)\in R$ is not feasible, since $n$ and $n'$ are of order $N!$. In this section we give an upper bound for $\MinOv_N(p,p')$ in terms of
the size $N$ and the degree distributions $p(k),p'(k)$. The analogous problem for $\MaxOv_N(p,p')$ will be the matter of Sect.~\ref{s4}.

We need to recall the standard \emph{configuration model algorithm} \cite{beca,boll,mr}
to generate a random network with a given degree distribution and size. We will use the following fast and efficient version of the algorithm. Let $K$ be a
degree set and let $(k_1,k_2,\ldots,k_N)$ be any degree sequence obtained by labeling the elements of $K$. In particular, $2L:=\sum k_i$ is even. Now take a vector $X$ of
length $2L$ containing $k_1$ times the integer 1 in the first $k_1$ entries, $k_2$ times the integer 2 in the following $k_2$ entries, etc. Each entry
$v$ of $X$ represents a single stub (or semi-link) attached at the node labeled as $v$. Then, take a random permutation of the entries of $X$ to get a new
array $Y$. Finally, read the contents of $Y$ in order, interpreting each pair of consecutive entries $v,w$ as a link between the nodes $v$ and $w$. For
an example, take $N=6$ and consider the degree distribution $p(k)$ defined by $p(1)=p(3)=1/6$, $p(2)=4/6$ and $p(k)=0$ for $k\ne1,2,3$. The corresponding
degree set is $\{1,2,2,2,2,3\}$. Take $(1,2,2,2,2,3)$ as degree sequence. Then, $X=(1,2,2,3,3,4,4,5,5,6,6,6)$. Now we permute $X$ at random,
obtaining $Y=(3,4,5,1,6,3,6,2,4,5,2,6)$. The set of links of the obtained network is $\{\{3,4\},\{5,1\},\{6,3\},\{6,2\},\{4,5\},\{2,6\}\}$. Observe that the link
$\{6,2\}$ appears twice. In general, the configuration model algorithm gives \emph{multigraphs} rather than graphs. It is well known, however, that
the fraction of self-loops and multi-links over the total number of links goes to 0 when $N\to\infty$ \cite{nsw}.

It seems natural to expect that the overlap between two networks of respective degree distributions $p(k),p'(k)$ and size $N$ generated via the
configuration model algorithm is very small. When the respective mean degrees are small with respect to the total size $N$ this turns out to be true. To prove
this fact, we need to estimate the probability that two given nodes are connected in a random network generated via the configuration model algorithm. So,
let $G$ be a network of $N$ nodes, $L$ links and degree distribution $p(k)$. Assume that $G$ has been obtained by means of the configuration model algorithm
starting with a degree sequence $(k_1,k_2,\ldots,k_N)$. Take at random any pair $\{i,j\}$ of nodes with $k_i\le k_j$. Next we estimate the probability
$p_{ij}$ that the network $G$ contains the link $\{i,j\}$. This probability is given by the quotient $a/b$, where $b$ is the total number of rearrangements $Y$
of the vector $X$ (here we are using the notation introduced in the definition of the configuration model) and $a$ is the number of such rearrangements having
at least two consecutive entries $i,j$ (or $j,i$) in places $Y_n,Y_{n+1}$ for $n=1,3,5,\ldots,2L-1$. We have that
\begin{equation}\label{b}
b=\displaystyle\frac{(2L)!}{k_1!k_2!\cdots k_N!}.
\end{equation}
Let us compute $a$. For $l=1,2,\ldots,L$, let $Y^l$ be the set of rearrangements $Y$ containing the entries $i,j$ (or $j,i$) in places
$Y_{2l-1},Y_{2l}$. Then, $a=|Y^1\cup Y^2\cup\ldots\cup Y^L|$. By the inclusion-exclusion principle, $a=a_1-a_2+\ldots+(-1)^{k_i-1}a_{k_i}$, where $a_l$
is the sum of the cardinalities of all intersections of $l$ sets in $Y^1,Y^2,\ldots,Y^L$. A simple combinatorial argument yields that
\[ a_l=\displaystyle\frac{{L \choose l}2^l(2L-2l)!}{k_1!k_2!\cdots k_{i-1}!(k_i-l)!k_{i+1}!\cdots k_{j-1}!(k_j-l)!k_j!\cdots k_N!}\mbox{ for $l\le k_i$}, \]
while $a_l=0$ for $k_i<l\le L$. Using the previous expression and the inclusion-exclusion principle we get that
\[ a=\displaystyle\frac{1}{k_1!k_2!\cdots k_{i-1}!k_{i+1}!\cdots k_{j-1}!(k_{j+1})!\cdots k_N!}\sum_{l=1}^{k_i}(-1)^{l-1}{L \choose l}2^l
\frac{(2L-2l)!}{(k_i-l)!(k_j-l)!}. \]
Taking it all into account, we get the following result.
\begin{theorem}\label{pij}
Let $G$ be a random network of $L$ links and $N$ nodes with degree sequence $(k_1,k_2,\ldots,k_N)$, generated via the configuration model
algorithm. Let $\{i,j\}$ be any pair of nodes with $k_i\le k_j$. Then, the probability that $G$ contains the link $\{i,j\}$ is
\[ p_{ij}=\displaystyle\frac{L!k_i!k_j!}{(2L)!}\sum_{l=1}^{k_i}(-1)^{l-1}2^l\frac{(2L-2l)!}{l!(L-l)!(k_i-l)!(k_j-l)!}. \]
\end{theorem}
The expression given by Theorem~\ref{pij} is too complex to be used to estimate the expected overlap between two random networks. Instead, we will
use the following standard approximation for the probability $p_{ij}$ \cite{dgm,nsw}:
\begin{equation}\label{pijapprox}
p_{ij}\approx \frac{k_ik_j}{2L-1}.
\end{equation}
This formula can be obtained from the proof of Theorem~\ref{pij} after replacing $a$ simply by $a_1$ (here we are using the notation of the proof).
The approximation \eqref{pijapprox} is good enough only when $k_i$ and $k_j$ are small with respect to $L$, in particular when we consider networks
with bounded mean degree and large size $N$, which is the case for most modeling applications. However, in general \eqref{pijapprox} can significantly
differ from the exact formula given by Theorem~\ref{pij}.

Using the approximation \eqref{pijapprox} we show that the expected overlap between two random networks generated via the configuration model is
very small, regardless of the particular distributions $p(k),p'(k)$, as the next result states.

\begin{theorem}\label{minov}
Let $p(k),p'(k)$ be two degree distributions with respective means $\chull{k}$ and $\chull{k'}$. Let $G,G'$ be two networks of $N$ nodes and
degree distributions $p(k)$ and $p'(k)$ generated via the configuration model algorithm. Assume that $N$ is big enough with respect
to $\chull{k}$ and $\chull{k'}$ in such a way that the approximation \eqref{pijapprox} holds. Then, the expected overlap between $G$ and $G'$ can
be approximated by
\[ \Ov(G,G')\approx\displaystyle\frac{\chull{k}\chull{k'}}{N(\chull{k}+\chull{k'})-\chull{k}\chull{k'}}. \]
\end{theorem}
\begin{proof}
Let $L,L'$ be the number of links of $G$ and $G'$ respectively. Assume that $G$ and $G'$ have been generated via the configuration model algorithm
starting with respective degree sequences $(k_1,k_2,\ldots,k_N)$ and $(k'_1,k'_2,\ldots,k'_N)$. Using the approximation \eqref{pijapprox} we can compute
the probability $p$ that two different nodes chosen at random are neighbors in $G$:
\begin{equation}\label{prob2}
\displaystyle p\approx\frac{1}{2L-1}\sum_{k_i,k_j}k_ip(k_i)k_jp(k_j) = \frac{\chull{k}^2}{2L-1}\approx\frac{\chull{k}^2}{2L}=\frac{\chull{k}}{N},
\end{equation}
where in the last expression we have used (S4). Now the expected overlap between $G$ and $G'$ can be computed as the probability that
two different nodes are connected in both $G$ and $G'$ over the probability that they are connected in $G\cup G'$ which, by virtue of \eqref{prob2}, is
\[ \displaystyle\frac{\chull{k}\chull{k'}/N^2}{1-\left(1-\frac{\chull{k}}{N}\right)\left(1-\frac{\chull{k'}}{N}\right)}. \]
\qed
\end{proof}

Theorem~\ref{minov} tells us that given $N$ and any two degree distributions $p(k),p'(k)$, the minimum overlap $\MinOv_N(p,p')$ is very close
to 0, at least when $N$ is big with respect to the expected values $\chull{k}$ and $\chull{k'}$. Of course, for small networks this is not true in general.

\section{An upper bound for the maximum overlap}\label{s4}
In this section we give an upper bound for $\MaxOv_N(p,p')$ in terms of the size $N$ and the degree distributions $p(k),p'(k)$.

Let $G,G'$ be two networks of $N$ nodes and empirical degree distributions $p(k),p'(k)$, with means $\chull{k}$ and $\chull{k'}$. Let $L$
and $L'$ be the number of links of $G$ and $G'$. If $E$ and $E'$ are the sets of links of $G$ and $G'$, then by definition
\begin{equation}\label{F(x)}
\Ov(G,G')=\displaystyle\frac{|E\cap E'|}{|E\cup E'|}=\frac{|E\cap E'|}{L+L'-|E\cap E'|}=\frac{x}{(\chull{k}+\chull{k'})\frac{N}{2}-x}=:F(x),
\end{equation}
where $x$ stands for $|E\cap E'|$ and in the last equality we have used (S4). Now observe that $F(x)$ is increasing as a function of $x$. Finally,
note that $x$ cannot be larger than $\min\{L,L'\}$. Assume without loss of generality that $L\le L'$. So, an upper bound for the maximum overlap permitted
between $G$ and $G'$ is obtained when replacing $x$ by $L=\chull{k}N/2$ in the previous expression, leading to
\[ \Ov(G,G')\le\displaystyle\frac{\chull{k}}{\chull{k'}}. \]
So, we have proved the following result.

\begin{lemma}\label{uppmax}
Let $p(k),p'(k)$ be two empirical degree distributions of $N$ nodes with respective means $\chull{k}$ and $\chull{k'}$. Then,
\[ MaxOv_N(p,p')\le\displaystyle\frac{\min\{\chull{k},\chull{k'}\}}{\max\{\chull{k},\chull{k'}\}}. \]
\end{lemma}

The upper bound in Lemma~\ref{uppmax} is too crude in general. In particular, one can have two completely different degree distributions with the same expected
values. In this situation, at least intuitively, there are important restrictions for the maximum value of the overlap, while the upper bound in
Lemma~\ref{uppmax} is 1. Let us see how to improve it.

Assume that we are given two degree sequences $D=(k_1,k_2,\ldots,k_N)$ and $D'=(k'_1,k'_2,\ldots,k'_N)$, with $\sum k_i=\chull{k}N=:2L$ and
$\sum k'_i=\chull{k'}N=:2L'$. Since $F(x)$ in \eqref{F(x)} is increasing in $x$, an upper bound for the overlap is obtained when replacing $x$ by the maximum
possible number of links of the intersection network. In the proof of Lemma~\ref{uppmax} this maximum was taken to be $\min\{L,L'\}$. To get a much better
estimate, look at a particular position $1\le i\le N$ of the degree sequences. It is clear that the intersection network cannot have more than
$\min\{k_i,k'_i\}$ links attached at node $i$. In consequence, the total number of links of the intersection network is at most
\[ L(D,D'):=\displaystyle\frac{1}{2}\sum_{i=1}^N \min\{k_i,k'_i\}. \]

The previous constant depends on the degree sequences $D$ and $D'$. Of course, reordering the elements of $D$ and $D'$ by means of permutations $\sigma,\tau$ we
get two degree sequences $\sigma(D),\tau(D')$ representing two networks with the same degree distribution. In consequence, we have the following result.

\begin{theorem}\label{Lpp}
Let $p(k),p'(k)$ be two empirical degree distributions of $N$ nodes and respective degree sets $K,K'$. Let $L_N(p,p'):=\max\{L(D,D')\}$, where the
maximum is taken over all pairs $D,D'$ of degree sequences obtained rearranging the elements of $K$ and $K'$ respectively. Then,
\[ \MaxOv_N(p,p')\le\displaystyle\frac{L_N(p,p')}{(\chull{k}+\chull{k'})\frac{N}{2}-L_N(p,p')}. \]
\end{theorem}

It is not easy to give a closed formula for $L_N(p,p')$ in terms of $N$, $p(k)$ and $p'(k)$. Alternatively, one could compute $L(D,D')$ for all possible
pairs $D,D'$ and select the maximum. This brute force algorithm is not feasible since the number of operations is about $N!$. Fortunately, there exists
an alternative and very fast algorithm to compute $L_N(p,p')$ that relies on the following simple lemma.

\begin{lemma}\label{rearra}
Let $(k_1,k_2,\ldots,k_N)$ and $(k'_1,k'_2,\ldots,k'_N)$ be two sequences of nonnegative numbers such that $k_1\le k_2\le\ldots\le k_N$. If there exists
a pair of indices $i<j$ such that $k'_i\ge k'_j$, then
\[ \min\{k_i,k'_i\}+\min\{k_j,k'_j\}\le\min\{k_i,k'_j\}+\min\{k_j,k'_i\}. \]
\end{lemma}
\begin{proof}
Since $k_i\le k_j$ and $k'_i\ge k'_j$, there are 6 cases to be considered to test the prescribed inequality:
\begin{itemize}
\item[$\bullet$] $k'_j\le k'_i\le k_i\le k_j$
\item[$\bullet$] $k'_j\le k_i\le k'_i\le k_j$
\item[$\bullet$] $k'_j\le k_i\le k_j\le k'_i$
\item[$\bullet$] $k_i\le k'_j\le k'_i\le k_j$
\item[$\bullet$] $k_i\le k'_j\le k_j\le k'_i$
\item[$\bullet$] $k_i\le k_j\le k'_j\le k'_i$.
\end{itemize}
It is trivial to check that the lemma holds in each case.
\qed
\end{proof}

As a consequence of Lemma~\ref{rearra} and Theorem~\ref{Lpp}, we get the following result.

\begin{theorem}\label{ordenades}
Let $p(k),p'(k)$ be two empirical degree distributions of $N$ nodes. Let $D=(k_1,k_2,\ldots,k_N)$ and $D'=(k'_1,k'_2,\ldots,k'_N)$ be the degree
sequences obtained by ordering increasingly the respective degree sets. Then,
\[ \MaxOv_N(p,p')\le\displaystyle\frac{\displaystyle\sum_{i=1}^N\min\{k_i,k'_i\}}{\displaystyle\sum_{i=1}^N\max\{k_i,k'_i\}}. \]
\end{theorem}
\begin{proof}
Lemma~\ref{rearra} states that if $S,S'$ are degree sequences fitting to $p(k)$ and $p'(k)$ such that $S$ is increasingly ordered and there is a
pair of entries $s'_i\ge s'_j$ of $S'$ with $i<j$, then if we swap both entries the obtained sequence $S''$ satisfies $L(S,S')\le L(S,S'')$.
Therefore, the maximum $L_N(p,p'):=\max\{L(S,S')\}$ is attained precisely in $L(D,D')$. Since, by definition, $L(D,D')=(1/2)\sum\min\{k_i,k'_i\}$,
Theorem~\ref{Lpp} and (S4) yield
\[ \MaxOv_N(p,p')\le\displaystyle\frac{(1/2)\sum_i\min\{k_i,k'_i\}}{(1/2)\big(\sum_i k_i+k'_i\big)-(1/2)\sum_i\min\{k_i,k'_i\}}. \]
Since $k_i+k'_i=\max\{k_i,k'_i\}+\min\{k_i,k'_i\}$, the theorem follows.
\qed
\end{proof}

Theorem~\ref{ordenades} allows us to design an efficient algorithm to compute an upper bound for the maximum overlap.
The algorithm takes as input the empirical distributions $p(k)$ and $p'(k)$. Sort increasingly the elements of the respective degree sets
to get sequences $D=(k_1,k_2,\ldots,k_N)$ and $D'=(k'_1,k'_2,\ldots,k'_N)$. Finally, return $\sum\min\{k_i,k'_i\}/\sum\max\{k_i,k'_i\}$.
In Table~\ref{compareup} we show the output of this algorithm for several pairs of empirical degree distributions, obtained via the configuration model from a
corresponding pair of (theoretical) distributions. In all cases, $N=10000$. Here "SF" stands for a scale-free network with $p(k) = C k^{-\gamma}$ with $\gamma=3$, minimum degree $m$, cut-off $k_c = m N^{1/2}$, and the normalization constant $C$, for which $\chull{k} \approx 2m$ (see Sect.~\ref{s6} for details). "Exponential" corresponds to $p(k)=(1/m)e^{1-k/m}$ with minimum degree $m$, for which $\chull{k}=2m$. "Poisson" corresponds to $p(k)=\lambda e^{-\lambda}/k!$ with $\lambda=\chull{k}$, and "Regular" stands for a random network for which all nodes have the same degree.

\begin{table}
\begin{center}
\begin{tabular}{|c|c|c|c|c|} \hline
 & Regular & Poisson & SF & Exponential \\ \hline
Regular & 0.7693 & 0.7508 & 0.6301 & 0.6654 \\ \hline
Poisson & 0.7552 & 0.7709 & 0.7221 & 0.7739 \\ \hline
SF & 0.5451 & 0.6000 & 0.7688 & 0.7023 \\ \hline
Exponential & 0.6330 & 0.7077 & 0.7715 & 0.7706 \\ \hline
\end{tabular}
\end{center}
\caption{Upper bounds for the maximum overlap permitted between pairs of empirical distributions according to Theorem~\ref{ordenades}.
In all cases, $N=10000$. For the left column distributions, $\langle k\rangle=20$ while, for the upper ones, $\langle k\rangle=26$.} \label{compareup}
\end{table}

\section{An algorithm to get a prescribed overlap}\label{s5}
Assume that we have generated two random networks $G(0),G'(0)$ of $N$ nodes using the configuration model. Let $p(k)$, $p'(k)$ be the corresponding empirical degree distributions. This section aims at designing an efficient algorithm to construct two networks $G,G'$ of $N$ nodes with respective degree distributions $p(k)$ and $p'(k)$ in such a way that $\Ov(G,G')\approx\alpha$, for any given $\MinOv_N(p,p')\le\alpha\le\MaxOv_N(p,p')$. Taking into account that, in view of Theorem~\ref{minov}, $\Ov(G(0),G'(0))\approx0$,
it seems natural to propose an algorithm that works as follows. At each time step $t\ge0$, modify the networks $G(t),G'(t)$ a little bit by performing a \emph{local operation}
(an operation involving few nodes and/or links) to obtain new networks $G(t+1),G'(t+1)$ with empirical degree distributions $p(k),p'(k)$ in such a way that $\Ov(G(t+1),G'(t+1))$
is slightly larger than $\Ov(G(t),G'(t))$. Repeat until the overlap is close to $\alpha$.

The kind of local operation that we will use in the scheme above is a \emph{cross rewiring operation} \cite{xs}, according to the following definition.
Let $G(t),G'(t)$ be two networks of $N$ nodes. A \emph{good pair in $G(t)$ with respect to $G'(t)$} is a pair of links $\{a,b\}$, $\{c,d\}$ in $G(t)$ satisfying
the following conditions:
\begin{enumerate}
\item $\{a,b\}$ and $\{c,d\}$ are not links in $G'(t)$
\item $\{a,c\}$ and $\{b,d\}$ are not links in $G(t)$
\item $\{a,c\}$ is a link in $G'(t)$.
\end{enumerate}
Analogously we define a \emph{good pair in $G'(t)$ with respect to $G(t)$} by interchanging the roles of $G(t)$ and $G'(t)$ in the previous definition.
Given a good pair $\{a,b\}$, $\{c,d\}$ in $G(t)$ with respect to $G'(t)$, the associated \emph{cross-rewiring operation} consists of replacing the links $\{a,b\}$ and $\{c,d\}$ in $G(t)$ by $\{a,c\}$ and $\{b,d\}$ to get a new network $G(t+1)$. Observe that $G(t)$ and $G(t+1)$ are in general different as non-labelled networks. However,
the degrees of the involved nodes $a,b,c,d$ are not modified after performing the cross-rewiring. In consequence, $G(t)$ and $G(t+1)$ have the same degree distribution. On the other hand, set $G'(t+1)=G'(t)$ and let $E(t)$, $E(t+1)$, $E'(t)$, $E'(t+1)$ be respectively the sets of links of $G(t)$, $G(t+1)$, $G'(t)$, $G'(t+1)$. Then, $|E'(t+1)|=|E'(t)|$ and, by the definition of the cross rewiring operation over a good pair, $|E(t+1)|=|E(t)|$. Moreover, by the definition of a good pair, either $|E(t+1)\cap E'(t+1)|=|E(t)\cap E'(t)|+1$ if $\{b,d\}$ is a link in $G'(t)$ or $|E(t+1)\cap E'(t+1)|=|E(t)\cap E'(t)|+2$ otherwise. Then, if we denote $\Ov(G(t),G'(t))$  and $\Ov(G(t+1),G'(t+1))$ by $\Ov(t)$ and $\Ov(t+1)$ respectively, a trivial computation yields that
\begin{equation}\label{mesov}
\Ov(t+1)=\Ov(t)+\frac{x\Ov(t)^2+2x\Ov(t)+x}{L-x-x\Ov(t)},
\end{equation}
where $x\in\{1,2\}$ and $L=|E(t)|+|E'(t)|$. In other words, the overlap after performing a cross rewiring operation in a good pair of links slightly (but strictly) increases.

From now on, let $\MinOv_N(p,p')\le\alpha\le\MaxOv_N(p,p')$ be the desired overlap coefficient. In view of what has been said, the following algorithm seems natural. Use the configuration model to construct two random networks $G(0),G'(0)$ of size $N$ and degree distributions $p(k),p'(k)$. The expected overlap is close to 0. Now, at each time step $t\ge0$, choose at random (if it exists) a good pair of links in $G(t)$ with respect to $G'(t)$. Perform a cross rewiring operation in $G(t)$ using such a pair, obtaining a new network $G(t+1)$. Set $G'(t+1):=G'(t)$. Then, $\Ov(G(t+1),G'(t+1))>\Ov(G(t),G'(t))$ by (\ref{mesov}). If $\Ov(G(t+1),G'(t+1))\ge\alpha$, set $G:=G(t+1)$, $G':=G'(t+1)$ and stop. Otherwise, proceed to the next time step.

A serious objection can be raised against the above algorithm as stated: there is no reason to expect that proceeding in this way we can reach values of the overlap close to $\MaxOv_N(p,p')$. It may well be that no more good pairs can be found to be rewired, long before reaching the desired overlap $\alpha$. To overcome this problem, we turn back to the proof of Theorem~\ref{ordenades}: the number of common links containing a given node $i$ cannot be larger than $\min\{k_i,k'_i\}$, where $k_i$ and $k'_i$ are the degrees of node $i$ in the respective networks. According to the proof of Theorem~\ref{ordenades}, to maximize the number of possible common links that will be obtained by performing a sequence of cross rewiring operations, it is enough to relabel the nodes increasingly with the degree. However, in doing this, we should make sure that the overlap between the original random networks does not change significantly (and remains, in consequence, close to 0). To support this claim, see Table~\ref{beaf}.

\begin{table}
\begin{center}
\begin{tabular}{|c|c|c|c|c|} \hline
 & Regular & Poisson & SF & Exponential \\ \hline
Regular & 0.00051 & 0.000469 & 0.000458 & 0.000478 \\
 & 0.00051 & 0.000509 & 0.000478 & 0.000487 \\ \hline
Poisson & & 0.000599 & 0.000383 & 0.000503 \\
 & & 0.000789 & 0.000836 & 0.000807 \\ \hline
SF & & & 0.000565 & 0.000652 \\
 & & & 0.002548 & 0.001671 \\ \hline
Exponential & & & & 0.000681 \\
 & & & & 0.001755 \\ \hline
\end{tabular}
\end{center}
\caption{The overlap between two random networks before and after relabeling the nodes increasingly with the degree. In all cases, $N=10000$ and $\langle k\rangle=10$.} \label{beaf}
\end{table}

So, let us consider the following \emph{CR Algorithm} (standing for \emph{Cross rewiring}), taking $p(k)$, $p'(k)$, $N$ and $\alpha$ as input:

\paragraph{{\bf CR Algorithm}}
\begin{enumerate}[(CR1)]
\item Use the configuration model to get two random networks $H(0),H'(0)$ of size $N$ and degree distributions $p(k),p'(k)$. Sort
increasingly the respective degree sequences $(k_1,\ldots,k_N)$ and $(k'_1,\ldots,k'_N)$. This corresponds to relabeling the nodes of both $H(0)$, $H'(0)$ to
get two networks $G(0)$, $G'(0)$ isomorphic to $H(0)$, $H'(0)$ respectively in such a way that $k_i\le k_j$ and $k'_i\le k'_j$ whenever $i<j$. The
overlap between $G(0)$ and $G'(0)$ is close to 0.
\end{enumerate}
\noindent At each time step $t\ge0$:
\begin{enumerate}[(CR2)]
\item 
Choose at random (if it exists) a good pair of links in $G(t)$ with respect to $G'(t)$. Perform a cross rewiring operation in $G(t)$ using such a pair, obtaining a new network
$G(t+1)$. Set $G'(t+1):=G'(t)$. Then, by (\ref{mesov}), $\Ov(G(t+1),G'(t+1))>\Ov(G(t),G'(t))$. If $\Ov(G(t+1),G'(t+1))\ge\alpha$, set $G:=G(t+1)$, $G':=G'(t+1)$ and stop. Otherwise,
proceed to the next time step.
\end{enumerate}

It is clear that after a finite number $t_0$ of steps the algorithm will stop, either because no good pairs are found or because the overlap between $G(t_0)$ and $G'(t_0)$ is very close to $\alpha$. In any case, the output of the algorithm is the pair of networks $G(t_0),G'(t_0)$. It is also clear that the algorithm admits some variants. For instance, the cross rewiring operations can be performed also over good pairs in $G'(t)$ with respect to $G(t)$. A natural question is whether in general the algorithm may halt forced by the condition that no good pairs are found, \emph{before} having reached a value of the overlap close to $\alpha$. This question can be reworded as follows: does the algorithm produce a value of the overlap coefficient close to $\MaxOv_N(p,p')$ when we execute it with $\alpha=\MaxOv_N(p,p')$? (observe that in this case the algorithm will stop if and only if no good pairs are found). In Table~\ref{guai} we show the maximum overlap generated using the CR Algorithm for several pairs of distributions, together with the upper bounds computed via the Theorem~\ref{ordenades}. In all cases, the obtained overlap is reasonably close to the theoretical maximum.

\section{Simulations}\label{s6}
We have performed a series of stochastic simulations with pairs of networks of size $N=10000$. In order to evaluate the accuracy of the analytical predictions depending on the network structure, we have chosen several
(theoretical) degree distributions $p(k),p'(k)$ for each layer. Once the size $N$ and the respective distributions $p(k)$ and $p'(k)$ are chosen, we proceed as follows:
\begin{enumerate}
\item Generate two random networks $A_0$ and $B_0$ with empirical degree distributions $p_A(k)\approx p(k)$ and $p_B(k)\approx p'(k)$ using the standard configuration model.
\item Use the corresponding degree sets and Theorem~\ref{ordenades} to estimate the maximum overlap coefficient $\alpha_{\max}$ (between any two networks distributed according to $p_A(k), p_B(k)$).
\end{enumerate}
Now we are ready to test the relevance of the layer overlap as a model parameter by choosing several values of $\alpha$ in the range $(0,\alpha_{\max})$. For any of such values, we use the CR Algorithm to construct two networks $A_\alpha$, $B_\alpha$, distributed according to $p_A(k), p_B(k)$, with an overlap coefficient very close to $\alpha$. With these ingredients we can simulate, using the standard Gillespie algorithm \cite{Gillespie}, the stochastic time evolution of the infection spread. In each case the initial number of infected nodes is set to 1000 (10$\%$ of the population size). The infected individuals are drawn from the whole population with the same probability $1/N$. In fact, for each pair $A_\alpha,B_\alpha$ we run 10 simulations with 10 different initial sets of infected nodes in order to average the outputs.

\begin{table}
\begin{center}
\begin{tabular}{|c|c|c|c|c|} \hline
 & Regular & Poisson & SF & Exponential \\ \hline
  & 1 & 0.739020 &  0.564752 & 0.448772 \\
Regular & 1 & 0.7761 & 0.6112 & 0.5004 \\ \hline
 & & 0.993035 & 0.654052 &  0.583859   \\
Poisson & & 0.994325 & 0.7180 & 0.6345 \\ \hline
 & & & 0.97987 &  0.665794 \\
SF & & & 0.98575 & 0.7095 \\ \hline
\end{tabular}
\end{center}
\caption{Maximum overlap generated using the CR Algorithm (first row) vs the maximum value permitted by
Theorem~\ref{ordenades} (second row). In all cases, $N=10000$ and $\langle k\rangle=10$. } \label{guai}
\end{table}

What we compare with the simulation outputs is the numerical integration of the model \eqref{SIS-HMFM2}, feeding it with the empirical degree distribution $p_A(k)$. We use the empirical distribution instead of the theoretical one $p(k)$ because, when the variance of $p(k)$ is large (highly heterogeneous networks), there can be noticeable differences among distinct finite samples of $p(k)$, in particular with respect to the values of the highest degrees which, as we will see, have a noticeable impact on the epidemic dynamics. To avoid degree-degree correlations within a layer due to the occurrence of very high degrees in the generated degree sequence, we have normalized the power-law distribution $p(k) = Ck^{-3}$ to have a minimum degree $m$ and a maximum degree given by the cut-off $k_c(N) = mN^{1/2}$, defined as the value of the degree above which one expects to find at most one node in the whole network. This expression of $k_c(N)$ coincides with the so-called structural cut-off for this exponent of the power law (see \cite{Catanzaro}), and leads to the normalization constant $C=(\gamma-1) m^{\gamma-1} N/(N-1)$ and an expected degree $\langle k \rangle =2mN/(N-1) \approx 2m$.

As initial condition $i_k(0)$ to integrate \eqref{SIS-HMFM2} and according to the procedure in the stochastic simulations, we consider that the same fraction of susceptible nodes becomes infected for any degree $k$. In particular, we take $i_k(0)=0.1 p_A(k)$ for all $k$, which amounts to a $10\%$ of initially infected nodes.

There is however a crucial remark on the simulation experiments. Recall that the lack of degree-degree correlations inside each layer was a basic assumption in the derivation of \eqref{SIS-HMFM2}. Therefore, to asses the goodness of the model we must make sure that all pairs of networks $A_\alpha,B_\alpha$ used in our simulations satisfy this assumption. It is reasonable to expect that the pairs of networks created via the CR Algorithm are uncorrelated, since:
\begin{enumerate}
\item The initial networks $G(0),G'(0)$ are randomly generated via the configuration model algorithm, which is known to produce uncorrelated networks.
\item A cross rewiring operation in a good pair of links $\{a,b\}$, $\{c,d\}$ increases (respectively, decreases) the global degree-degree correlation if the new links connect the two nodes with the smaller degrees and the two nodes with the larger degrees (respectively, if one of the new links connects the node with the largest degree to the node with lowest degree). But the rewiring criterion in the CR Algorithm is intended to increase the overlap coefficient and has nothing to do with the degrees of the four involved nodes. So, some reconnections will increase the global degree-degree correlation and some will decrease it, thus expecting an overall balance.
\end{enumerate}
To support this claim, we show in Table~\ref{degdeg} the standard Pearson coefficient $r$ for each layer, computed from the two random variables defined by the degrees of the nodes at both ends of randomly chosen links \cite{newm}. Values of $r$ close to $-1$ (respectively 1) account for dissortative (resp. assortative) networks, while values close to 0 correspond to uncorrelated networks.

\begin{table}
\begin{center}
\begin{tabular}{|c|c|c|c|} \hline
 & $\alpha=0.2$ & $\alpha=0.4$ & $\alpha=0.6$ \\ \hline
 Poisson & -0.003353 & -0.003353 & 0.022467 \\
 SF & 0.016494 & 0.057019 & 0.085278 \\ \hline
 SF & -0.007762 & -0.007762 & -0.007762 \\
 Exponential & -0.010041 & -0.070960 & -0.070971 \\ \hline
 Poisson & -0.003353 & -0.003353 & -0.003353 \\
 Exponential & 0.040715 & 0.088316 &  0.139976 \\ \hline
\end{tabular}
\end{center}
\caption{Pearson coefficient to measure the degree-degree correlations in each layer for several pairs of networks obtained from the CR Algorithm. In all cases, $N=10000$ and $\langle k\rangle=10$. } \label{degdeg}
\end{table}

We have performed two series of experiments addressed to illustrate the influence of the two factors appearing in \eqref{pB|A}, both of which are related to the topology of the layers. The first factor accounts for the difference in link density between both layers (measured by the ratio of their mean degrees), whereas the second one is an increasing function of the overlap coefficient $\alpha$. First, we consider that both layers have exponential degree distributions but different minimum degrees and, hence, different mean degrees, $\langle k_A \rangle$ and $\langle k_B \rangle$. Figure~\ref{Exp-Exp} shows, for $\alpha=$ 0.1 (left panels) and 0.6 (right panels), the prevalence of the epidemic when $\langle k_A \rangle = 30 > \langle k_B \rangle = 20$ (top panels) and when $\langle k_A \rangle = 20 < \langle k_B \rangle = 30$ (bottom  panels). From this figure it follows that the epidemic will be better contained when (an important part of) layer A can be embedded in layer B, which is only possible when $\langle k_A \rangle < \langle k_B \rangle$. Such an embedding is clearly not possible when the number of links is much larger in layer A than in layer B ($\langle k_A \rangle > \langle k_B \rangle$).

\begin{figure}[ht]
\begin{tabular}{ll}
\hspace{-1cm}
	\includegraphics[scale=0.35]{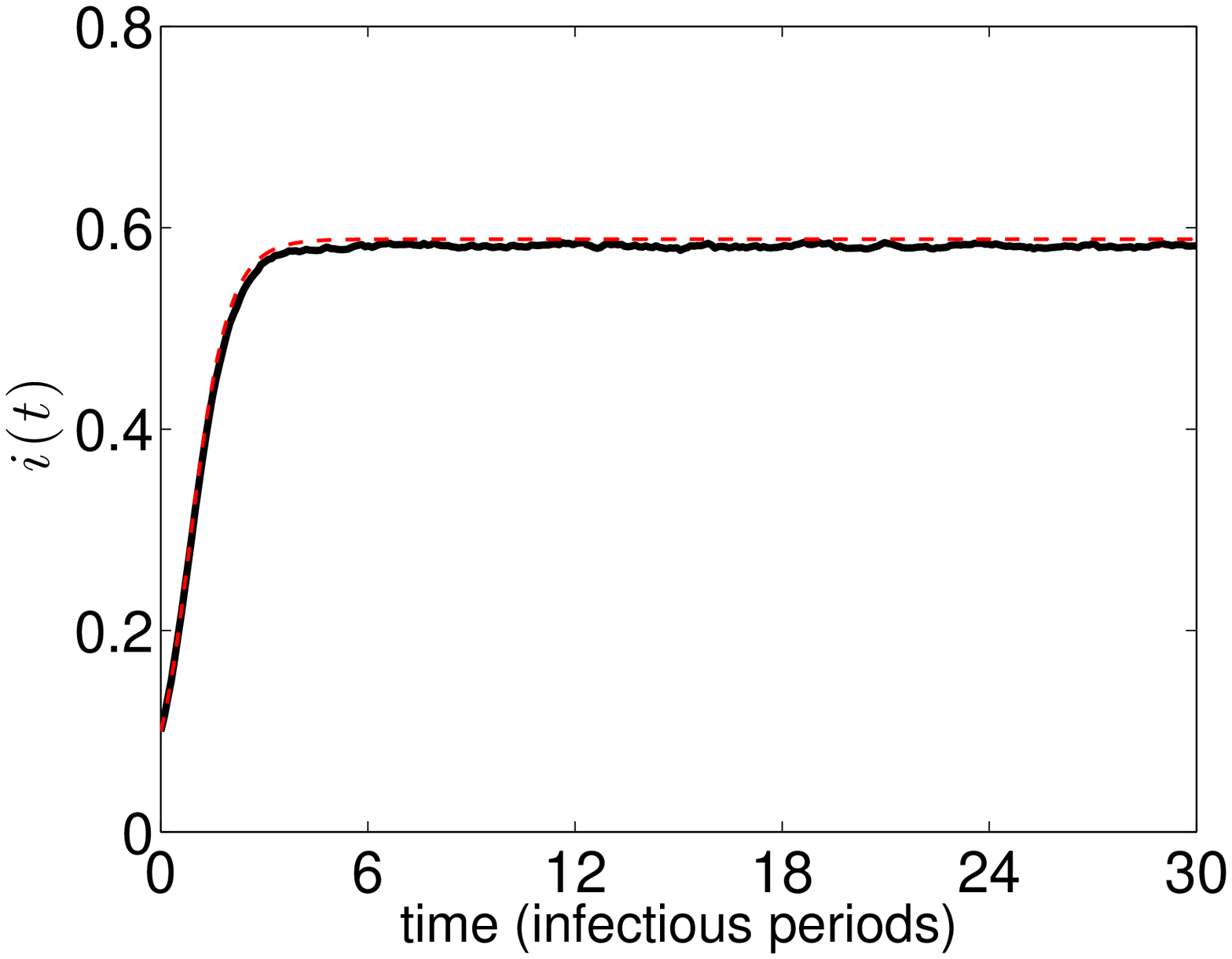}
&
	\includegraphics[scale=0.35]{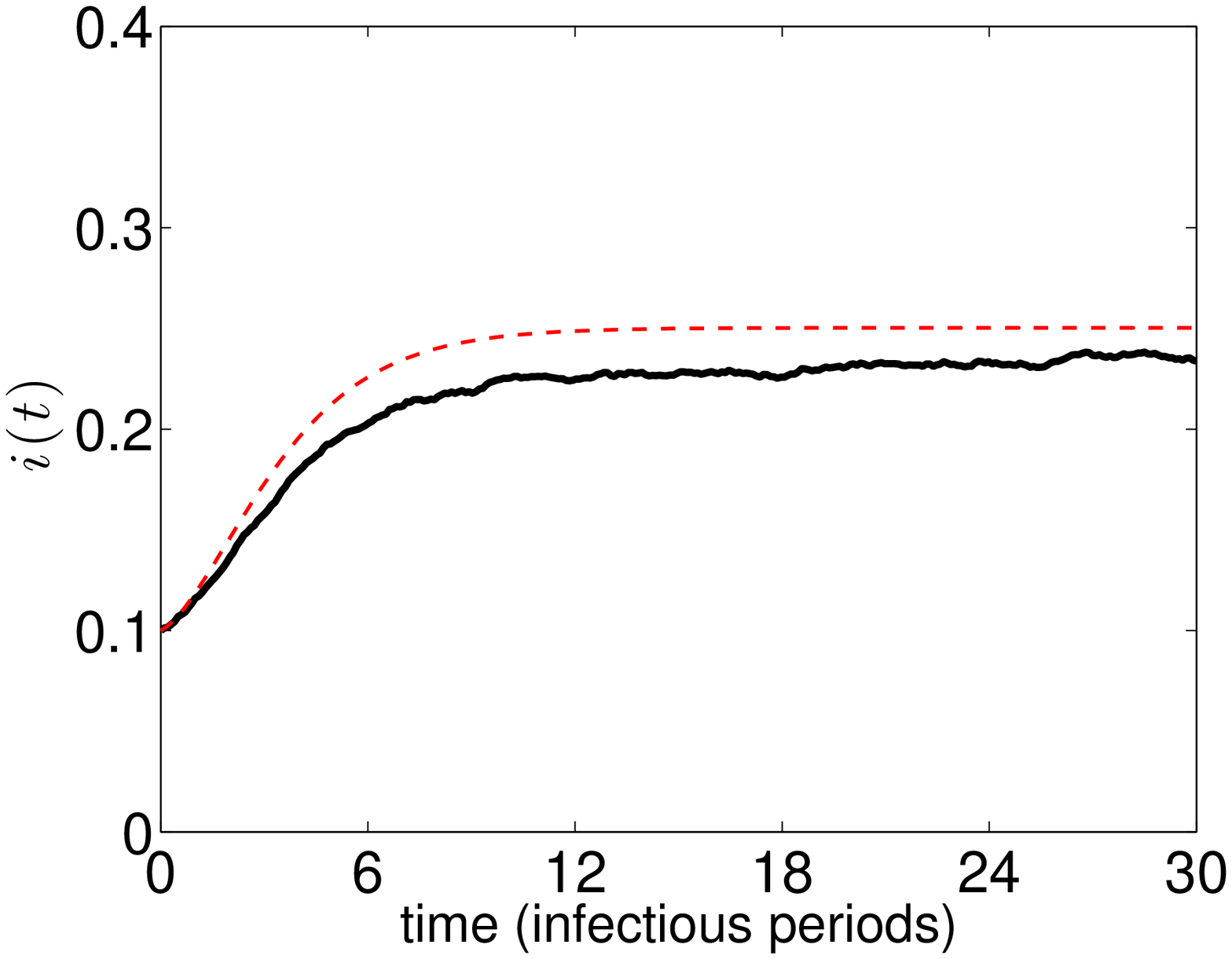}
\\
\hspace{-1cm}
	\includegraphics[scale=0.35]{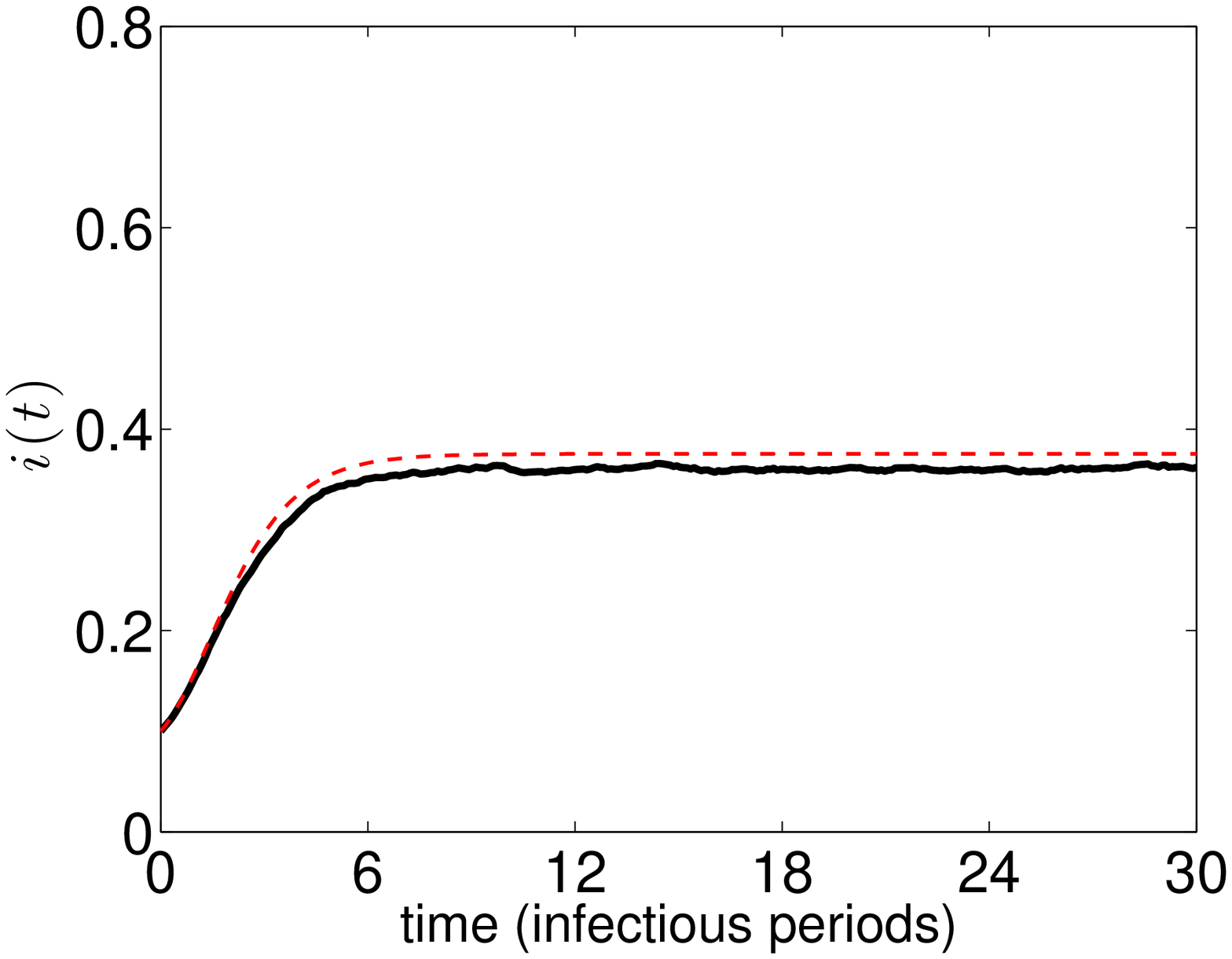}
&
	\includegraphics[scale=0.35]{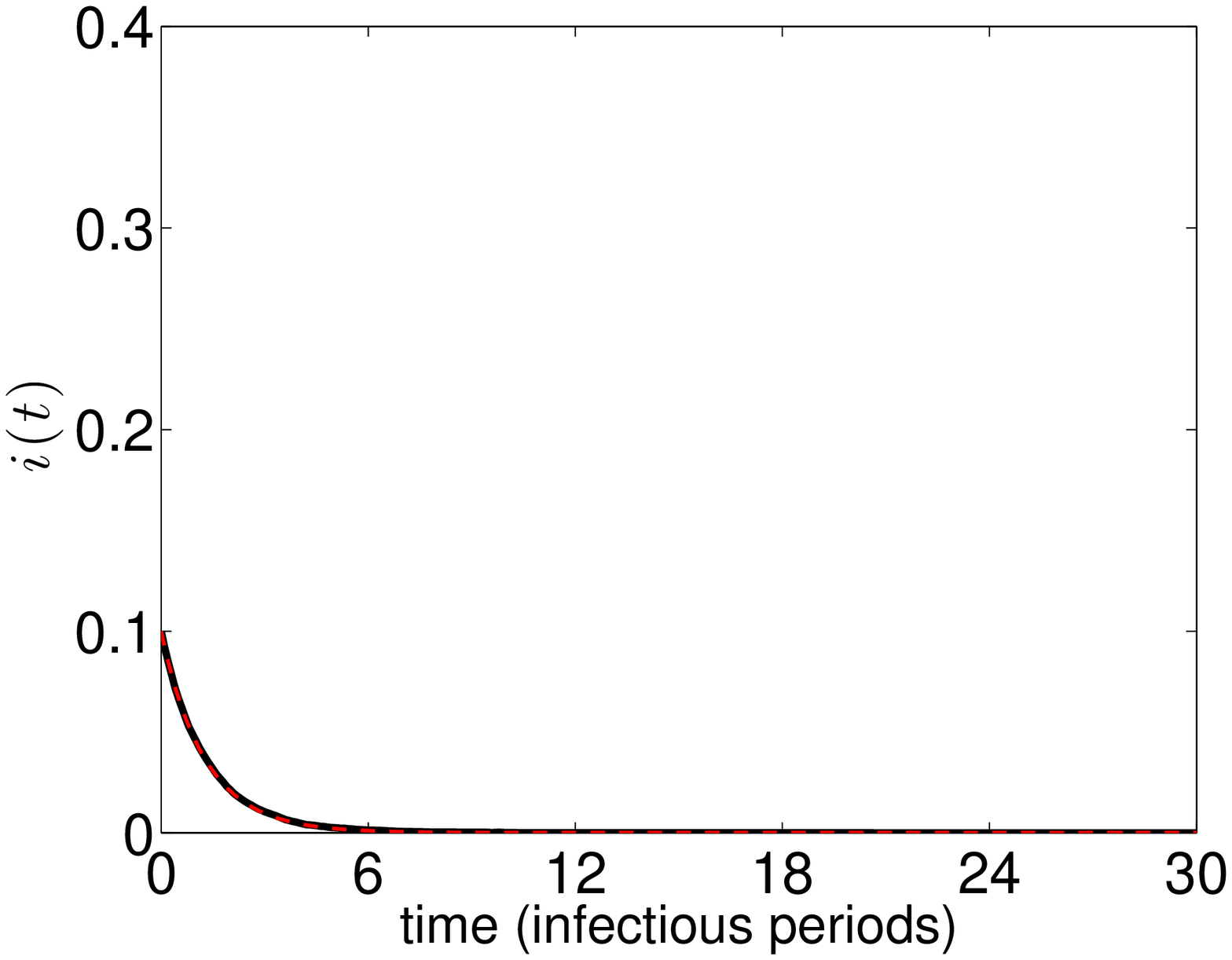}
\end{tabular}
\caption{Fraction of infectious nodes averaged over 10 runs of stochastic simulations carried out on two-layered exponential random networks of size $N=10000$ for $\alpha=$ 0.1 (left panels) and $\alpha=$ 0.6 (right panels). Top panels: $\langle k_A \rangle =30$ and $\langle k_B \rangle=20$. Bottom panels: $\langle k_A \rangle =20$ and $\langle k_B \rangle=30$. Dashed line shows the prevalence $\left( \sum_k i_k \right)$ predicted by the SIS model \eqref{SIS-HMFM2}. Initial fraction of infected nodes: 10$\%$. Parameters: $\mu=1$, $\beta=0.1$, $\beta_c=0.005$.
\label{Exp-Exp}}
\end{figure}

Next, we compare network layers with the same expected number of links (same mean degrees) but different network topologies. The aim is to see how a non-uniform overlap makes the epidemic dynamics depart from the model predictions. Keeping the same heterogeneous degree distribution in layer A, we vary the degree heterogeneity in layer B by considering Poisson, exponential and power-law degree distributions. To make the differences more noticeable, we take a parameter combination leading to epidemic extinction according to model \eqref{SIS-HMFM2}. As Figure~\ref{SF-Others} shows, when the variance of degrees in layer B is low (top panel) the nodes with the highest degrees in layer A have a much lower fraction of common links than those with low degrees once the CR algorithm has been applied. This means the violation the hypothesis of a uniform overlap between layers, and allows a higher transmission of the infection which leads to a (low) prevalence of the disease, instead of the epidemic die-out predicted by the model. As the variance of degrees in layer B increases (middle and bottom panels), the disagreement between simulations and the model prediction decreases. In fact, in the bottom panel the epidemic extinction is also observed in the simulations because layer B has a degree sequence generated from the same power-law distribution as the one used to generate the degree sequence of layer A and, hence, a higher uniformity in the overlap is achieved.

\begin{figure}[ht!]
\begin{tabular}{l}
	\includegraphics[scale=0.35]{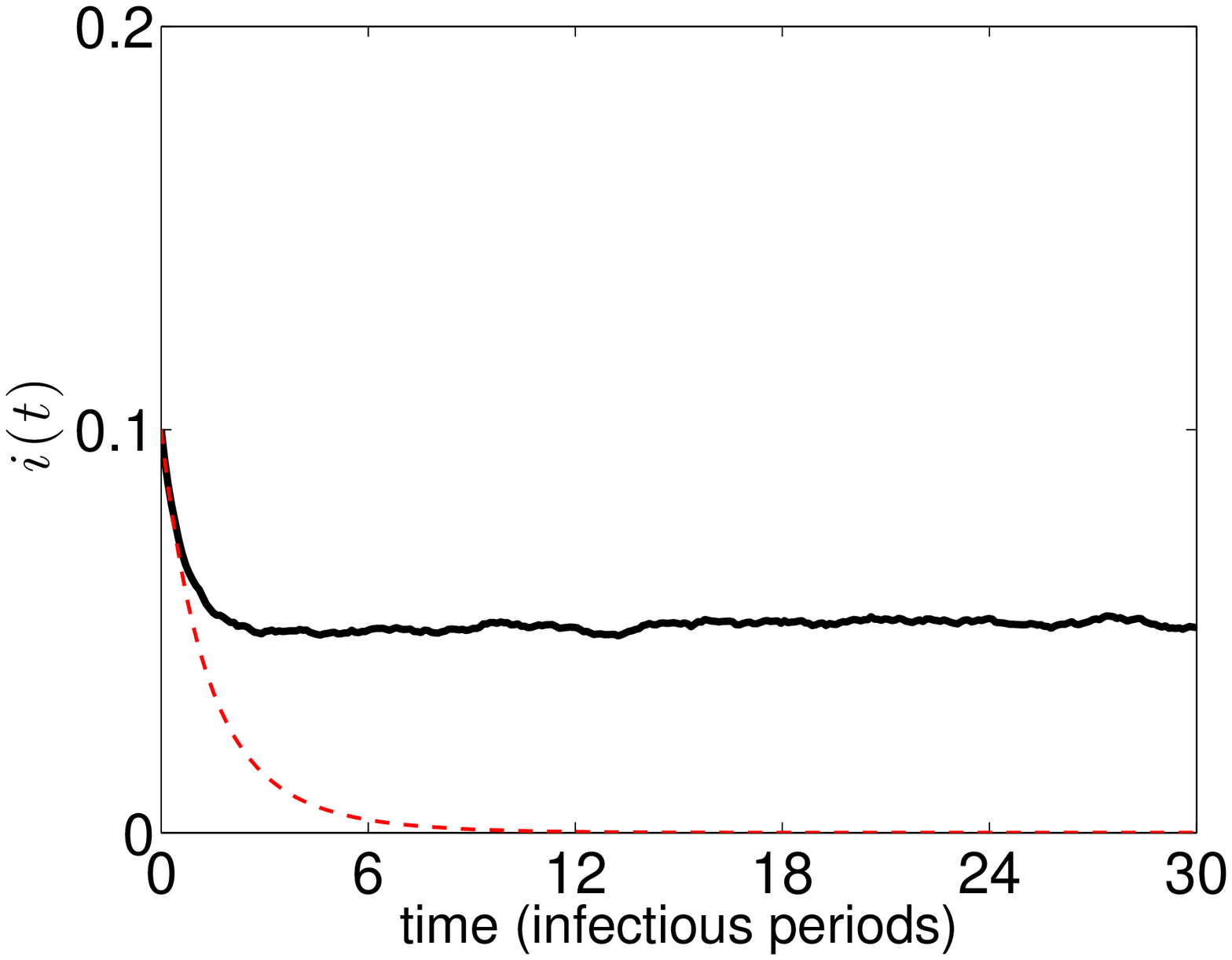}
\\
	\includegraphics[scale=0.35]{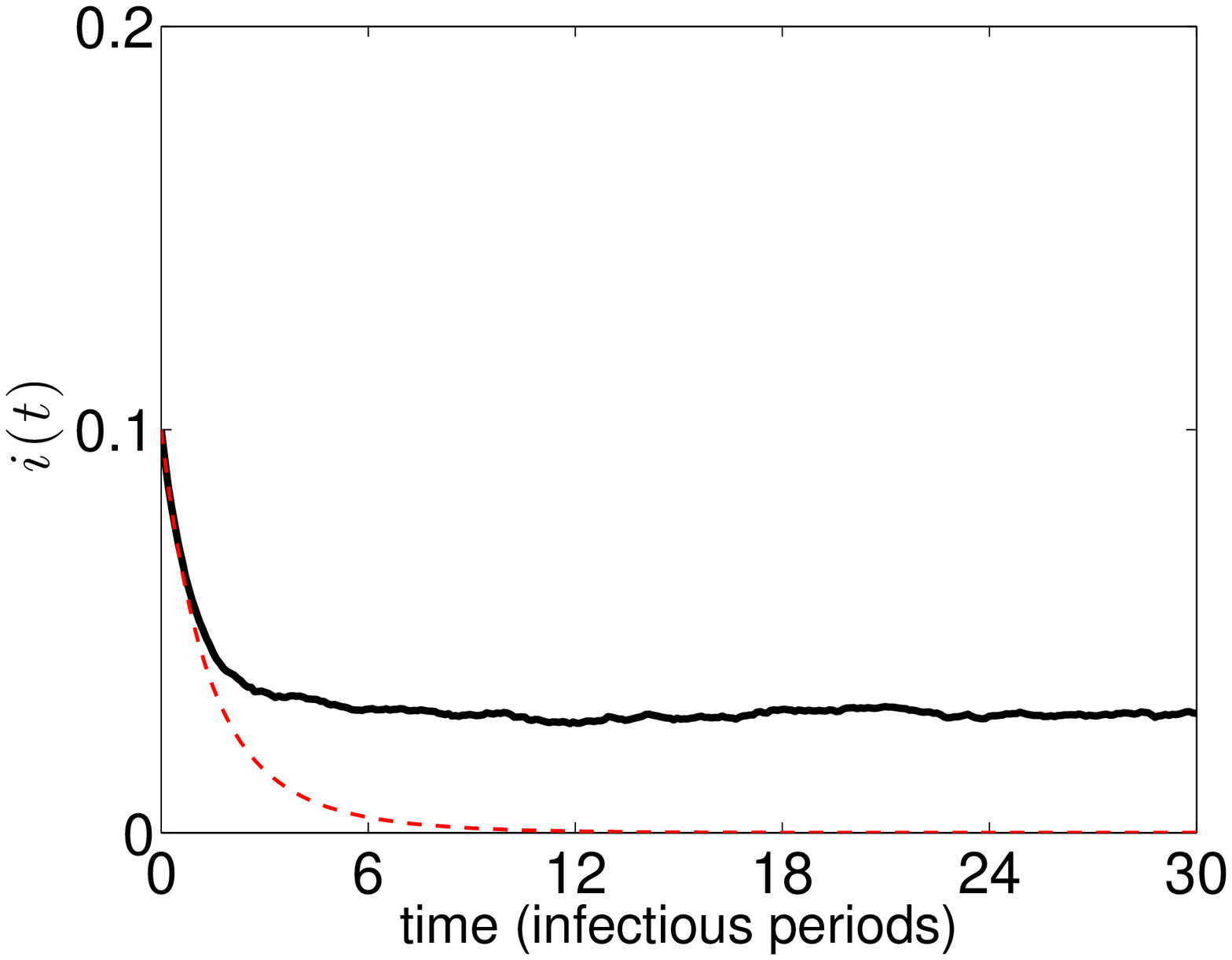}
\\
	\includegraphics[scale=0.35]{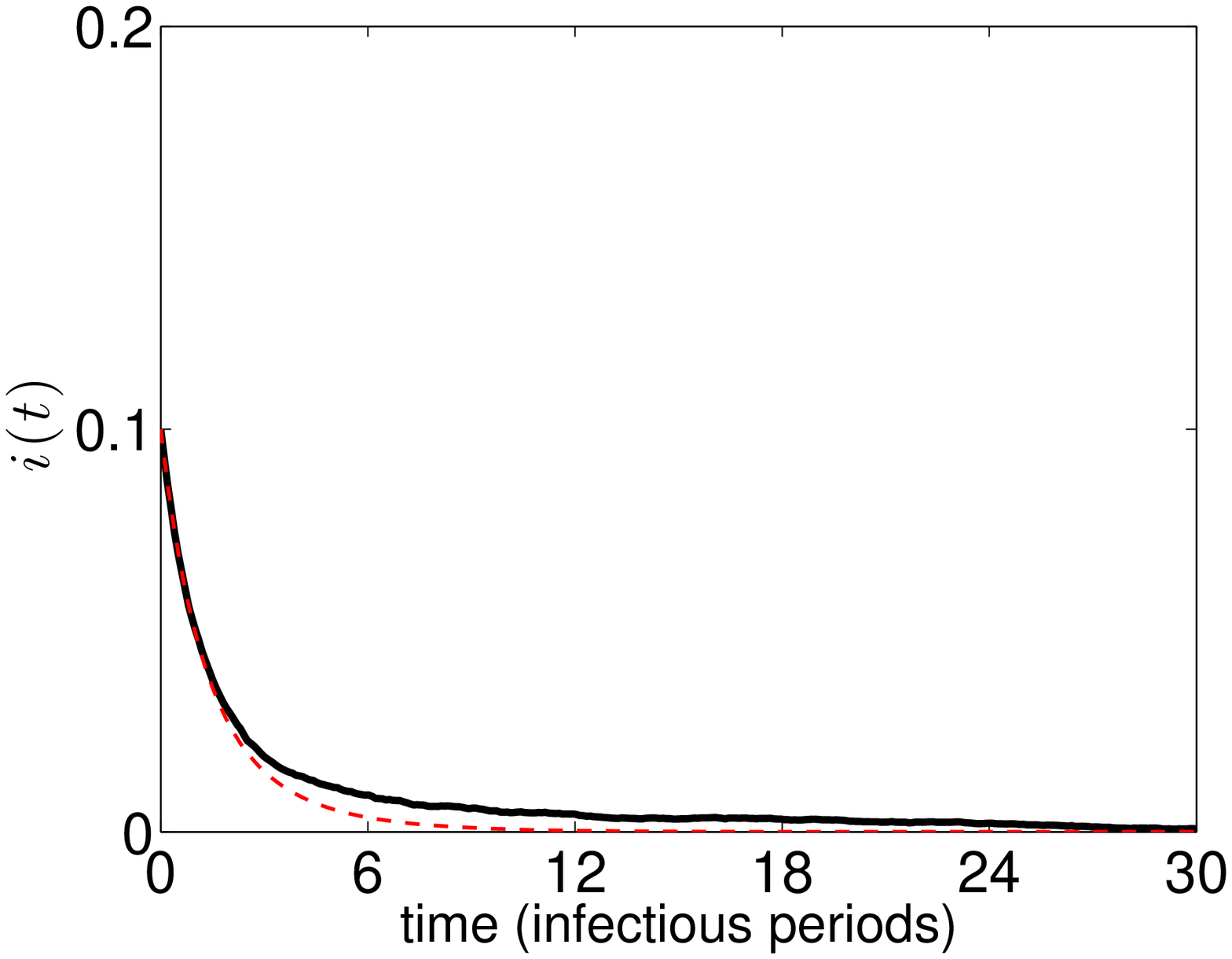}
\end{tabular}
\caption{Fraction of infectious nodes averaged over 10 runs of stochastic simulations carried out on a two-layered network of size $N=10000$
for $\alpha=$ 0.6. In all cases layer A has a power-law degree distribution ($p(k) \sim k^{-3}$) with $k_{\min}=5$ ($\langle k_A \rangle =10$), whereas layer B has Poisson (top), exponential (middle) and power-law (bottom) degree distribution with the same mean degree ($\langle k_B \rangle =10$). Dashed line shows the prevalence $\left( \sum_k i_k \right)$ predicted by the SIS model \eqref{SIS-HMFM2}. Initial fraction of infected nodes: 10$\%$. Parameters: $\mu=1$, $\beta=0.1$, and $\beta_c=0.005$.
\label{SF-Others}}
\end{figure}

\section{Discussion}
We have proposed a cross-rewiring algorithm to create and control the overlap between two networks with prescribed degree sets. The wider range of permitted overlap coefficients, from 0 to values very close to the theoretical upper bound given by Theorem~\ref{ordenades}, is obtained by cross-rewiring networks whose nodes have been labelled according to their rank in the ordered degree sequences, suggesting that the overlap coefficient and the inter-layer degree-degree correlation can be quite independent from each other. This algorithm allows to check the predictions of a mean-field SIS model with awareness dissemination in a host population, where the routes of propagation for the infectious agent and awareness are embedded into a two-layer network.

A key ingredient of the model is the probability $p_{B|A}$ that a randomly chosen link of layer $A$ connects two nodes that are also connected in layer $B$, i.e., the probability that an  $A$-link is a common link. Its expression, given by \eqref{pB|A}, shows that, as one could expect, it increases with the overlap between the layers but, moreover, it is also a linear increasing function of the ratio $\langle k_B \rangle / \langle k_A \rangle$, which measures the difference in the number of links of each layer, $L_A$ and $L_B$. In particular, if $\langle k_A \rangle > \langle k_B \rangle$, Lemma~\ref{uppmax} says that $\alpha < \langle k_B \rangle / \langle k_A \rangle = L_B / L_A$, and, so, $p_{B|A} < L_B / L_A$. This inequality simply reflects the fact that layer $A$ cannot be embedded in layer $B$ (since $p_{B|A} < 1$). Conversely, if $\langle k_A \rangle \le \langle k_B \rangle$, then $\alpha \le \langle k_A \rangle / \langle k_B \rangle$, and \eqref{pB|A} implies that $p_{B|A} \le 1$. This result agrees with what one would expect from the definition of $p_{B|A}$ because, when all the $A$-links are common links, $p_{B|A} = 1$ even for $\alpha < 1$, but such an embedding is only possible when $L_A \le L_B$. The effect of these asymmetric roles played by each mean degree on the epidemic progression is illustrated in Fig.~\ref{Exp-Exp} for networks with the same type of degree distribution but different mean degrees. Clearly, in this example the epidemic spread is only contained when the mean degree of the layer $B$, over which awareness dissemination occurs, is higher than that of layer $A$ and, moreover, the overlap between layers is high enough (right bottom panel).

A basic assumption in the derivation of the model is the uniform distribution of the overlap over the set of nodes. This means that those nodes with high degrees in layer $A$ have the same fraction of overlapped links that those with lower degrees. Of course, this will not be the case when there is a large asymmetry between the degree distributions of each layer. One can observe the differences when layer $A$, the one over which physical contacts occur, has a power-law degree distribution whereas dissemination layer $B$ has a Poisson degree distribution. When both degree distributions have the same mean degree, those nodes with the highest degrees in layer $A$ only have a small fraction of overlapped links because of the low variance of the Poisson distribution. This amounts to an underestimation of the epidemic prevalence by the mean-field SIS model \eqref{SIS-HMFM2} since those nodes acting as a superspreaders in layer A have proportionally much less contacts with a low transmission rate (see top panel in Fig.~\ref{SF-Others}). In contrast, by increasing the variance of the degree distribution of layer $B$, disease transmission is reduced and the epidemic evolution is closer to the one predicted by the the mean-field model (see bottom panel in Fig.~\ref{SF-Others} where layer $B$ has the same power-law degree distribution as layer $A$).

In general, when the mean-field assumptions are met, stochastic simulations confirm that the proposed SIS model is suitable for modelling two interacting contagious processes like epidemic spreading and awareness dissemination. In particular, due to the nature of their interaction, the model predicts a decreasing relationship between $R_0$ and the overlap coefficient $\alpha$ (see Fig.~\ref{R0-SIS}). Moreover, although the analytical prediction of the mean-field model is not accurate close to the epidemic threshold $R_0(\alpha)=1$, the behaviour of the prevalence with network overlap shows a good agreement with stochastic simulations when the overlap coefficient is not so close to its critical value. With this respect, it would be interesting to consider which relationships between overlap and epidemic thresholds follow for more general epidemic network models as those considered in \cite{KJS,Sahneh14,SSVM}, which are based on the adjacency matrix of each network layer and allow for both degree-degree correlations within and between layers, and a non-uniform overlap between layers.

Finally, note that a similar mean-field approach for modelling epidemic spreading on single heterogeneous networks was adopted in \cite{Boguna,PS-V} using, as state variable, the fraction $\rho_k$ of nodes of degree $k$ that are infectious. The connection between this approach and the one traditionally used in epidemiology is given by the relationship between the state variables, namely, $i_k = I_k/N = I_k/N_k \cdot N_k/N = \rho_k p(k)$ (see \cite{May-Lloyd}). These works were more focussed on aspects of network topology and, in particular, the absence of epidemic threshold was proved in \cite{Boguna} for scale-free networks with degree-degree correlations, i.e., for networks with a mixing pattern such that $P(k'|k) \ne k'p(k')/\chull{k}$. Such a network-oriented approach offers an alternative way for analysing the impact of overlap on epidemic spreading on two-layer networks with non-proportionate mixing within each layer (see \cite{Saumell} for an extension of this formalism to interconnected networks). With this respect, the cross-rewiring algorithm used to generate overlapped networks with arbitrary degree distributions can be adapted to control the intra-layer degree-degree correlation during the process. This network attribute, however, will restrict the value of the maximum attainable overlap coefficient because it reduces the number of "good pairs" as long as correlations within each layer are preserved. Indeed, the dependence between correlations and the maximum attainable overlap constitutes an interesting topic for future work.

\end{document}